\documentclass{article}

\usepackage[margin=3.2cm]{geometry}   

\usepackage[nolist]{acronym}
\usepackage{graphicx}
\usepackage[dvipsnames]{xcolor}
\usepackage{subcaption}
\usepackage{epsfig}
 
\usepackage{amssymb}
\usepackage{cancel, comment, nicefrac}
\usepackage{mathtools}

\newtheorem{theorem}{Theorem}
\newtheorem{lemma}{Lemma}
\newtheorem{proposition}{Proposition}
\newtheorem{assumption}{Assumption}
\newtheorem{remark}{Remark}
\newenvironment{proof}{\noindent\textbf{Proof.}}{\hfill$\square$\par}

\newcommand{\cL}{\mathcal{L}}

\newcommand{\cQ}{\mathcal{Q}}

\newcommand{\N}{\mathbb{N}}

\newcommand{\R}{\mathbb{R}}

\usepackage{hyperref}

\begin{document}

\begin{acronym}
    \acro{RoA}{Region of Attraction}
    \acro{FH-RoA}{Finite Horizon Region of Attraction}
    \acro{SoS}{Sum of Squares}
    \acro{LMI}{Linear Matrix Inequality}
    \acro{SDP}{Semi Definite Programming}
    \acro{LP}{Linear Program}
\end{acronym}

\title{\bf Convex computation of regions of attraction from data using Sums-of-Squares programming}

\author{Oumayma Khattabi, Matteo Tacchi-Bénard, Sorin Olaru}

\maketitle

\begin{abstract}
This paper focuses on the analysis of the \ac{RoA} for unknown autonomous dynamical systems. A data-driven approach based on the moment-\ac{SoS} hierarchy is proposed, enabling novel \ac{RoA} outer approximations despite the reduced information on the dynamics. The main contribution consists in bypassing the system model and, hence, the recurring constraint on its polynomial structure. Numerical experiments showcase the influence of data on learned approximating sets, highlighting the potential of this method.

\end{abstract}

\section{Introduction}

Stability analysis of dynamical systems is one of the fundamental pillars of control theory. It is traditionally based on the system model, leveraging approaches such as Lyapunov's direct method~\cite{lyapunovGeneralProblemStability1992}, LaSalle's invariance principle~\cite{la2012stability}, or the comparison principle in dissipativity theory~\cite{george2025comparison}. While relying on a system model gives strong theoretical guarantees, defining a precise model is becoming harder with the growing complexity of modern systems. In recent years, data-driven methods increased in popularity due to the abundance of data, increased computational power, and the need to adapt to nonlinear real systems. As a result, many stability analysis approaches have been developed or adapted to prove stability directly from data.

Model-based and data-driven techniques have contributed significantly to the stability analysis of nonlinear systems. In both settings, several works stand out. For instance, ~\cite{julian-parametrization-1999} uses linear programming to learn piecewise affine Lyapunov functions   for uncertain systems; while~\cite{TacchiTAC25} uses second-order cone programming to approximate such functions from data. A related data-based approach is the continuous piecewise affine method proposed in~\cite{hafstein2016computing}.

Neural networks have also become practical tools in stability analysis. In~\cite{kolter2019learning}, neural Lyapunov candidates are identified jointly with the system model from data to establish stability. They are further employed in~\cite{min2023data} to generate control laws constrained by Lyapunov functions learned online, in~\cite{grune2021overcoming} while mitigating the curse of dimensionality in Lyapunov function identification, and in~\cite{kim2024estimation} to compute constraint-admissible invariant sets.

Another approach is the use of polynomials: in~\cite{martin-data-driven-2024}, dissipativity is analysed from noisy data using piecewise Taylor approximations. On the model-based side, \ac{SoS} programming has been used to find control Lyapunov functions in~\cite{tan2004searching}, estimate the \ac{RoA} via polynomial Lyapunov functions in~\cite{biswas2023region}, and compute inner and outer \ac{FH-RoA} approximations in~\cite{KordaTAC14, KordaNOLCOS13}.

While \ac{SoS}-based methods provide strong theoretical guarantees for stability analysis, they are typically limited to model-based polynomial systems. In contrast, data-driven approaches enable the study of more general nonlinear dynamics. Some works use data to approach the unknown part of the model, and \ac{SoS} methods to learn Lyapunov candidates~\cite{colbert2018using} or controllers with guarantees~\cite{han2022sum}. This work, in contrast, examines the feasibility and effectiveness of a fully data-driven approach under minimal model information.

The remainder of the paper is organised as follows. In section~\ref{sec:data-driven}, the theoretical framework of the data-driven concept is provided, as well as a trivial example in $\R$ to illustrate. Section~\ref{sec:opt} details the optimisation problem based on data and its complementary case of inner approximations. Numerical results are presented in section~\ref{sec:num}, along with an analysis of the computational difficulty and a proposed solution. An appendix is provided at the end of the paper to explain the moment-SoS hierarchy and prove the convergence of the optimisation problem.

\noindent{\bf Notation:} Let $n \in \N$. $\|.\|$ is the Euclidean norm.
For $X\subset\R^n$, $\partial X$ is the boundary of $X$, $X^c$ its complement and $\mathrm{vol(X)}$, its $n$-dimensional volume. We define $[n] \triangleq \{1,\ldots,n\}$. 
$\R[x]^{n}$ is the set of $n$-dimensional vectors of real polynomials in $x$. $C^n(X)$ is the set of functions on $X$ with continuous derivatives up to order $n$. $W^{n,\infty}(X)$ is the set of functions with bounded derivatives up to order $n$ on $X$.

\section{Data-driven region of attraction}\label{sec:data-driven}

\subsection{Problem statement}

We consider a dynamical system:
\begin{equation}\label{f0}
    \dot x=f(x)
\end{equation}
defined by an unknown vector field $f:\mathbb{R}^n\rightarrow \mathbb{R}^n$. {We seek the set $X_0(f)$ of initial states whose trajectories stay in $X \subset \mathbb{R}^n$ over $[0,T]$, and reach $X_T \subset X$ at time $T>0$. We call $X_0(f)$ the \ac{FH-RoA}, defined for a horizon $T$, an admissible set $X$, and a target set $X_T$.} To underline its dependence on the original model (1), the set $X_0$ will be denoted as such:
\begin{equation} \label{eq:modelROA}
    X_0(f) \triangleq \left\{ x_0 \in X \middle| \begin{array}{l}
        \exists x(\cdot) \in C^1([0,T])^n \quad \text{s.t.} \\[0pt]
        x(0) = x_0, \quad x(T) \in X_T, \\[0pt] 
        \forall t \in [0,T], \quad x(t) \in X \\[0pt]
        \text{and} \quad \dot{x}(t) = f(x(t))
    \end{array} \right\}.
\end{equation}

Although unknown, $f$ is accessible as follows:
\begin{assumption} \label{asm: data}
    The function $f$ is Lipschitz continuous, and the following information is available:
    \begin{itemize}
        \item an upper bound $M>0$ on its Lipschitz constant, 
        \item a finite and fixed sample of {noise-free} evaluation points:   $D = \{(x_i,y_i = f(x_i))\}_{i\in [N]} \subset \R^{2n}. $
    \end{itemize}
\end{assumption}

\begin{remark}
    Bounding arguments are common in data-driven analysis: partial derivative bounds {estimated from data}~\cite{martin-data-driven-2024,makdesi2021efficient}, Lipschitz bounds \cite{TacchiTAC25}, Hilbert norm bounds~\cite{maddalena2021deterministic}, etc.
\end{remark}

\begin{remark} \label{rem: uncertainty}
    Assumption~\ref{asm: data} implies that $\forall x \in X$, $f(x)$ lies within a state-dependent semi-algebraic uncertainty set:
    \begin{align}
        F_{D}(x) 
        & = \left\{y \in \R^n \; \middle | \begin{array}{l}
            \forall i \in [N], \\[0pt]
            \|y-y_i\|^2 \leq M^2 \|x - x_i\|^2 
        \end{array}\right\}. \label{eq:uncertainty}
    \end{align}
\end{remark}

{\begin{remark}
    In this work, the data are assumed noise-free. Noisy datasets can still be incorporated by assuming $y_i = f(x_i) + z_i$, with $\|z_i\|\leq\delta$ a bounded disturbance. This simply enlarges the uncertainty set by the disturbance bound $\delta$: $\|z_i\|^2 \leq \delta^2 \Rightarrow \|y-y_i-z_i\|^2 \leq M^2\|x-x_i\|^2$.
\end{remark}}

\begin{remark} \label{rem: difficulty}
    Even when $f$ is known, and $X$ and $X_T$ are simple sets, exact computation of $X_0(f)$ is rarely tractable. For instance, if $f \in \R[x]^n$, and $X$ and $X_T$ are semi-algebraic, the moment-SoS hierarchy enables the computation of certified inner (\cite{KordaNOLCOS13}) and outer (\cite{KordaTAC14}) approximations of $X_0(f)$ via convex optimisation. The present contribution extends these works to the setting where $f$ can only be accessed through the dataset $D$, and the Lipschitz bound $M$.
\end{remark}

To compute certified approximations of $X_0(f)$, we intend to bridge the gap between model-based moment-SoS frameworks in~\cite{KordaNOLCOS13, KordaTAC14}, and the data-based approach in~\cite{TacchiTAC25}. As a result, we make the following assumption on the structure of the admissible and target sets $X$ and $X_T$:

\begin{assumption} \label{asm: algebra}
    $X$ and $X_T$ are compact and there exist {$n_X, n_T \in \N$}, vectors $g_X \in \R[x]^{n_X}$ and $g_T \in \R[x]^{n_T}$, 
    s.t.:
    \begin{subequations}
        \begin{equation}
            X = \{x \in \R^n \mid g_X(x) \geq 0\}
        \end{equation}
        \begin{equation}
            X_T = \{x \in \R^n \mid g_T(x) \geq 0\}
        \end{equation}
    \end{subequations}
    where vector inequalities are considered component-wise.
\end{assumption}

\subsection{Finite time region of attraction}

We {revisit} the formulation proposed in~\cite{KordaTAC14} for outer \ac{FH-RoA} approximation. The authors 
{over-approximate} the horizon $T$ \ac{FH-RoA} of the target $X_T$ 
for {a (model-based) polynomial control system with polynomial input constraints, seen as a differential inclusion. The resulting \ac{FH-RoA} is defined identically to~\eqref{eq:modelROA}, with the additional requirement that there \textit{exists} a control law generating the converging trajectory.} 
%
In this contribution, we work with uncontrolled but unknown dynamics. The key observation is that~\cite{KordaTAC14} does not require the input $u$ to be a control {input}: the methodology applies if $u$ is replaced with an unknown disturbance $y$, and our framework corresponds to 
{$\dot{x} = y$} with $y \in F_D(x)$ as described in~\eqref{eq:uncertainty}. Hence, one can formulate the model uncertainty as a fixed implicit semi-algebraic set inclusion:
\begin{equation} \label{eq:implicit}
    \hspace{-0.5em}\begin{array}{l}
        (x,\dot{x}) \in \Gamma_D \qquad \text{with} \\[0pt]
        \Gamma_D \triangleq \left\{(x,y) \in \R^{2n} \middle| \begin{array}{l}
        g_X(x) \geq 0 \; \text{and} \; y \in F_D(x)
        \end{array} \right\}.
    \end{array}
\end{equation}

This way, the data-based \ac{FH-RoA} can be expressed as follows:
\begin{equation} \label{eq:dataROA}
    X_0 \triangleq \left\{x_0 \in X \middle| \begin{array}{l}
        \exists x(\cdot) \in W^{1,\infty}([0,T])^n \ \ \text{s.t.} \\[0pt]
        x(0) = x_0, \ x(T) \in X_T \ \ \& \\[0pt]
        \forall t \in [0,T], \ (x(t),\dot{x}(t)) \in \Gamma_D
    \end{array}\right\}.
\end{equation}
In words, $X_0$ is the set of all initial conditions of trajectories $x(\cdot)$, such that $x(T) \in X_T$ and at all times $t \in [0,T]$, $x(t) \in X$, and $\dot{x}(t) \in F_D(x(t))$. 

\begin{remark} \label{rem:best}
    The description~\eqref{eq:dataROA} of $X_0$ requires \textit{the existence of at least one} dynamical system that maps $x(0)=x_0\in X_0$ to $x(T)\in X_T$, subject to state constraints and compliant with the data. We refer to this case as the \textit{best-case} \ac{FH-RoA}, which provides an outer approximation of the \ac{FH-RoA} of the real system. In contrast, the worst-case \ac{FH-RoA} denotes the set $X_0^\star$ corresponding to \textit{all} dynamical systems compliant with the data that map $x(0)=x_0\in X_0^\star$ to $x(T)\in X_T$. This yields an inner approximation of the \ac{FH-RoA} of the real system. 
    We show that the worst-case inner approximation can be obtained from the best-case outer approximation of $X_0^c$, the complementary set of $X_0$.
\end{remark}

\begin{remark}
    From a different perspective, the trajectories initiated in $X_0$ are \emph{viable} (\cite{aubin2011viability}) over the horizon $T$, with target set $X_T$, according to the uncertainty that can be inferred from the available data.
\end{remark}

\begin{remark}
    It is well known that a control input and an unknown disturbance are mathematically equivalent; the only difference is that the control value can be chosen. The interplay between the existence of a control steering the system to the target and the requirement of reaching the target under all possible perturbations is standard in $H^\infty$ control. The originality of our approach lies in applying this observation to the notions of inner and outer approximations of the \ac{FH-RoA} in a data-driven framework.
\end{remark}

{\vspace{-1.5em}\begin{remark}
    Our approach differs from~\cite{KordaTAC14} in four aspects: $(i)$ it does not require the underlying dynamics $f$ to be polynomial, they can be any Lipschitz continuous vector field; $(ii)$ it requires no model at all, and is instead fully data-driven; $(iii)$ it accounts for the more general \textit{coupled} constraint $y \in F_D(x)$ instead of the \textit{decoupled} $x \in X, u \in U$; $(iv)$ it can provide either \textit{best-case} outer approximations or \textit{worst-case} inner approximations.
\end{remark}}

\subsection{Illustration on a toy example}

The concept is explained using the following 1D system defined for $x \in X=[-1,1]\subset \R$ by:
\begin{equation}\label{eq:1D_example_sys}
    \dot{x} = 
        \begin{cases}
        2x(x^2-0.5^2) \quad \text{if} \quad |x| \leq 0.5 \\[0pt]
        x - 0.5 \quad \text{if} \quad x \geq 0.5 \\[0pt]
        x + 0.5 \quad \text{if} \quad x \leq -0.5
        \end{cases}
\end{equation}

 For illustration purposes, the data is generated by the model $f$ given in~\eqref{eq:1D_example_sys} that is assumed unknown, and we show how the dynamics of~\eqref{eq:1D_example_sys} can be analysed with only data. Using the following dataset with three points: $D_1 = \{(-1;-0.5),(0;0),(1;0.5)\}$ and the Lipschitz constant $M=1$ of the system, we plot the area $F_{D_1}(X)$ where all the possible function values exist (see green region in figure~\ref{fig:1D-3p-example}). The largest possible \ac{FH-RoA} corresponds to the piecewise affine function constructed out of the Lipschitz inequality limits with values in quadrants II and IV. We call this function $f_{best}$, associated with the best-case RoA.
 \begin{figure}[htbp]
     \centering
     \includegraphics[width=0.67\linewidth]{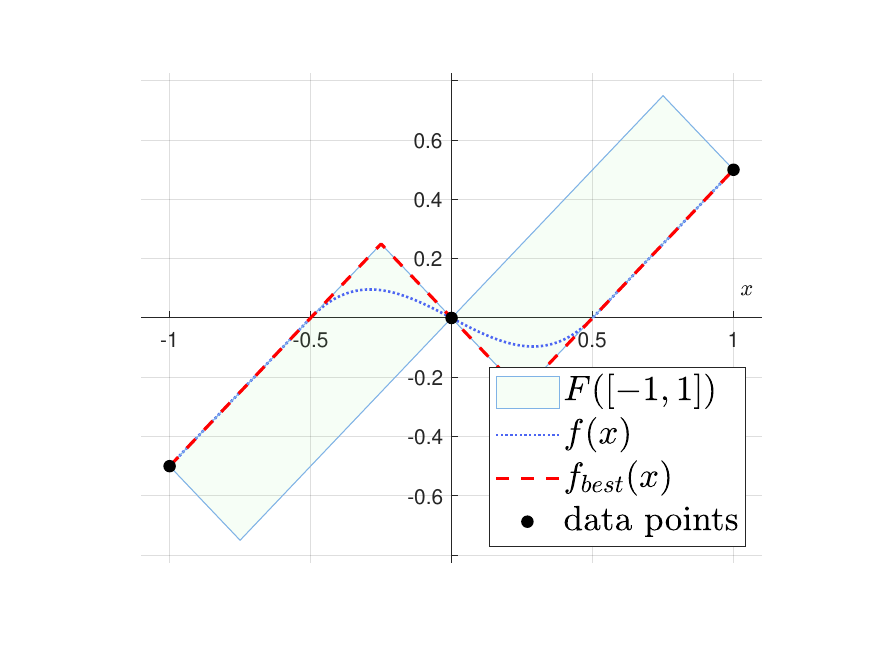}
     \caption{Toy example with 3 data points}
     \label{fig:1D-3p-example}
 \end{figure}

\vspace{-1em}
{Given the same dataset $D_1$ and $X_T = [-0.25,0.25]$ for $T=1s$, one can analytically compute $f_{best}$ as in figure~\ref{fig:1D-3p-func}, as well as its corresponding \ac{FH-RoA}. In this case, the \ac{FH-RoA} of system~\eqref{eq:1D_example_sys} for $T$ and $X_T$ is $X_0(f) = [-0.34,0.34]$, and the \ac{FH-RoA} of $f_{best}$ for the same parameters is $X_0(f_{best}) = [-0.408,0.408]$ (see figure~\ref{fig:1D-3p-RoA}), making it an outer approximation of $X_{0}(f)$.}

\begin{figure}[htbp]
    \centering
    \begin{subfigure}{.49\linewidth}
        \includegraphics[width=\linewidth]{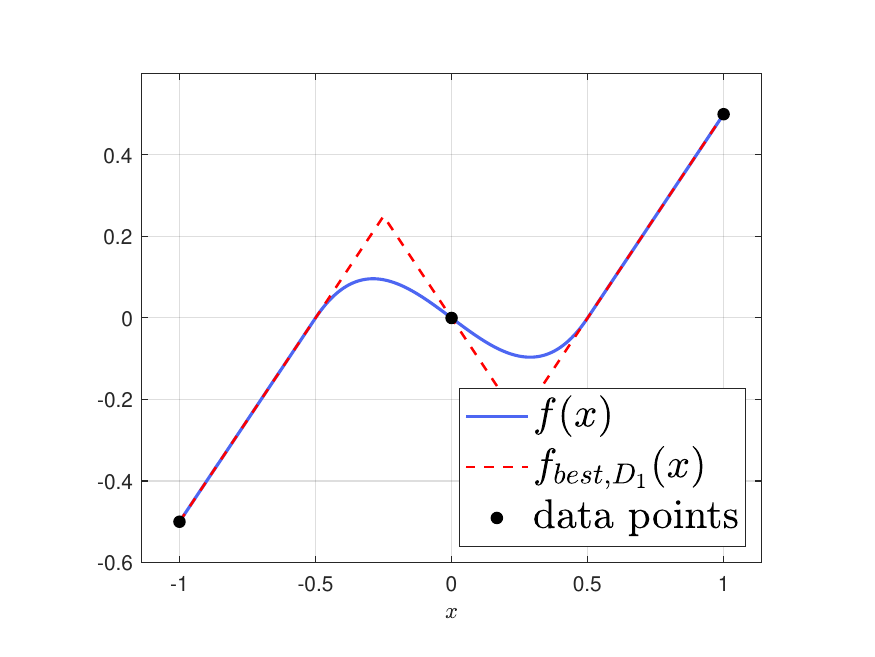}
        \caption{dataset $D_1$}
        \label{fig:1D-3p-func}
    \end{subfigure}
    \begin{subfigure}{.49\linewidth}
        \includegraphics[width=\linewidth]{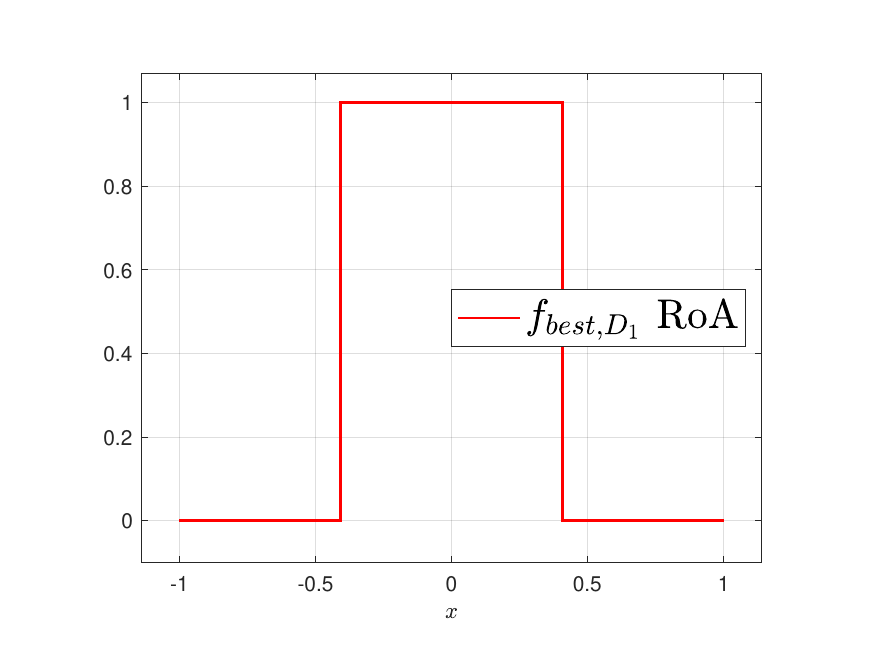}
        \caption{indicator function of the \ac{FH-RoA}}
        \label{fig:1D-3p-RoA}
    \end{subfigure}
    \caption{Identification of the \ac{FH-RoA} with 3 data points}
    \label{fig:1D-3p}
\end{figure}

Adding 2 points $D_2 = D_1 \cup \{(-0.3;0.096),(0.3;-0.096)\}$ imposes tighter constraints (see figure~\ref{fig:1D-5p-func}), reducing the feasible solution space (see figure~\ref{fig:1D-5p-RoA}). However, such an analytical approach is cumbersome even for scalar dynamics.

\begin{figure}[htbp]
    \centering
    \begin{subfigure}{.49\linewidth}
        \includegraphics[width=\linewidth]{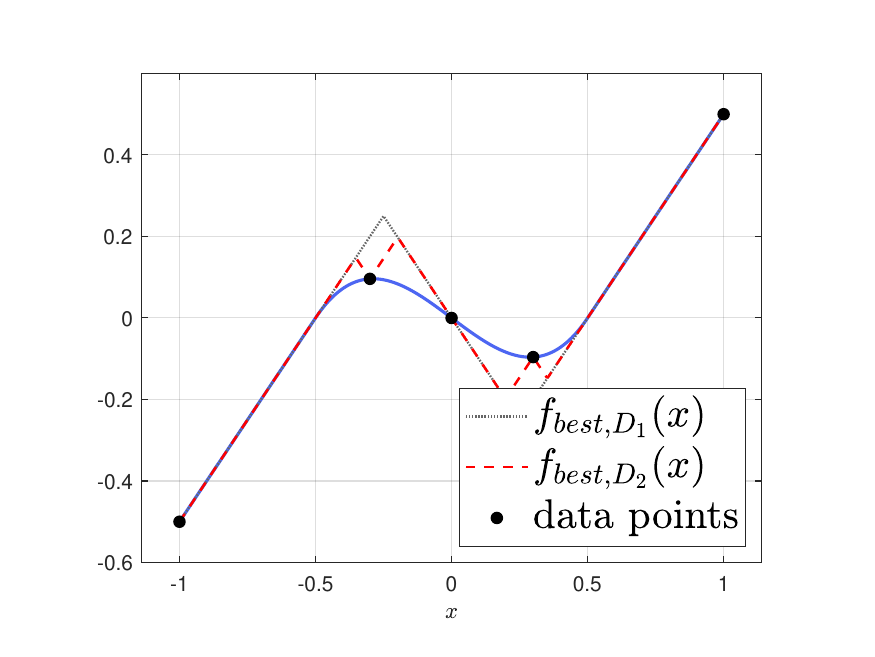}
        \caption{dataset $D_2$}
        \label{fig:1D-5p-func}
    \end{subfigure}
    \begin{subfigure}{.49\linewidth}
        \includegraphics[width=\linewidth]{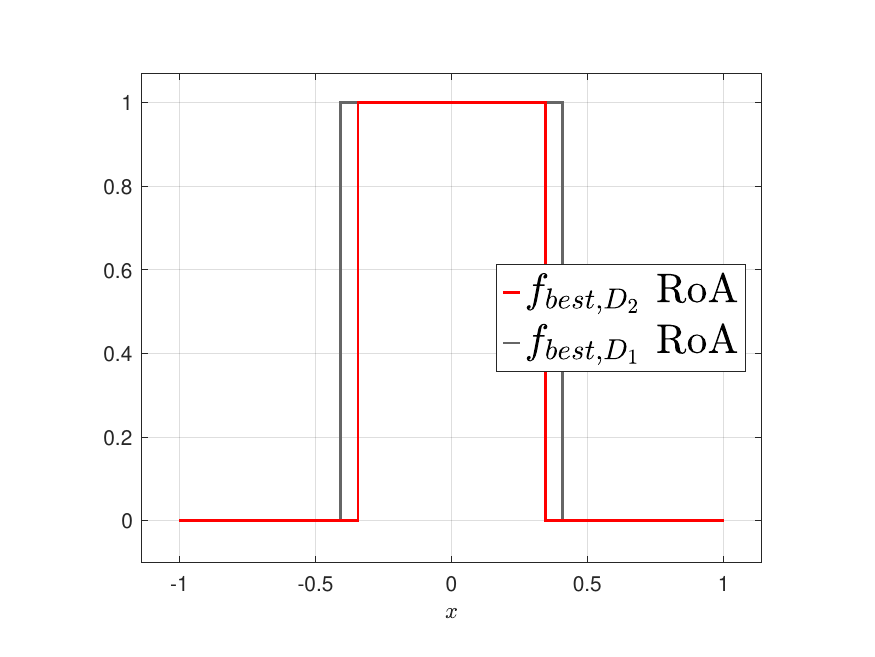}
        \caption{indicator function of the \ac{FH-RoA}}
        \label{fig:1D-5p-RoA}
    \end{subfigure}
    \caption{Identification of the \ac{FH-RoA} with 5 data points}
    \label{fig:1D-5p}
\end{figure}

\begin{remark} \label{rem:lipschitz}
    The description of $X_0$ in~\eqref{eq:dataROA} implicitly allows non-Lipschitz behaviour in the dynamics $\dot{x}$ as a function of $x$, requiring the Lipschitz condition to hold only on the data points (formally, $f_{best}$ in Fig.~\ref{fig:1D-3p-example} is allowed to jump within the green envelope). However, the toy example suggests that the $f_{best}$, whose \ac{FH-RoA} corresponds to the best-case \ac{FH-RoA}, always lies on the boundary of the uncertainty set $F_D(x)$, and is exactly $M$-Lipschitz continuous.
\end{remark}


\section{The optimisation framework}\label{sec:opt}

\subsection{Linear program on functions}

{We propose a generalisation of}~\cite{KordaTAC14} from explicit semi-algebraic differential inclusions (and polynomial control systems) to implicit semi-algebraic set inclusions (and systems with bounded uncertainty).

Let $\cL_y$ denote the operator acting on $v\in C^1([0,T]\times X)$, defined for $(t,x,y) \in [0,T]\times \Gamma_D$ by:

\begin{equation} \label{eq:Lie}
    \cL_y v(t,x,y) \, \triangleq
    \frac{\partial v}{\partial t}(t,x) + y^\top\frac{\partial v}{\partial x}(t,x)
\end{equation}
The \ac{LP} reads as follows:
\begin{subequations} \label{eq:LP}
\begin{align}
    \hspace{-1em} \underset{\substack{v \in C^1([0,T]\times X) \\[0pt] w \in C^0(X)}}{\mathrm{minimise}} & \int_X w(x) \; dx \\[0pt]
    \text{s.t.} \qquad & \cL_y v(t,x,y) \leq 0 & \hspace{-2em} \forall (t,x,y) \in [0,T]\times \Gamma_D \label{con: occupation} \\[0pt]
    & w(x) \geq 0 & \forall x \in X \label{con:wgeq} \\[0pt]
    & w(x) \geq v(0,x) + 1 & \forall x \in X \label{con:vw}\\[0pt]
    & v(T,x) \geq 0 & \forall x \in X_T \label{con:terminal}
\end{align}
\end{subequations}

\begin{remark} \label{rem: difference}
    The difference with~\cite{KordaTAC14} is that~\eqref{con: occupation} is defined on $[0,T]\times\Gamma_D$ instead of $[0,T]\times X\times U$ and controlled dynamics $\dot{x} = f(t,x,u)$ are replaced with unknown speed of variation $\dot{x} = y \in F_D(x)$.
\end{remark}

{\begin{lemma} \label{lem:cv}
    Let $(v,w)$ be feasible for problem~\eqref{eq:LP} and define 
    $$
        \hat{X}_0(w) \triangleq \{x \in X \mid w(x) \geq 1\}.
    $$
    Recalling the definition~\eqref{eq:dataROA} of $X_0$, it holds
        $X_0 \subset \hat{X}_0(w)$. 
    Moreover, considering a minimising sequence $(v_\epsilon, w_\epsilon)$ s.t.
    $$0 \leq \int w_\epsilon(x) \; dx - w^\star \leq \epsilon,$$
    where $w^\star$ is the optimal value of problem~\eqref{eq:LP}, it holds
    \begin{equation} \label{eq:cv}
    \mathrm{vol}(\hat{X}_0(w_\epsilon) \setminus X_0) \underset{\epsilon \to 0}{\longrightarrow} 0.
    \end{equation}
\end{lemma}
\begin{proof}
    Almost identical to the ones in ~\cite{KordaTAC14, KordaNOLCOS13, KordaSIOPT14, OustryLCSS19}: Let $x_0 \in X_0$ and consider an admissible trajectory $x(t)$ such that $x(0) = x_0$, $x(T) \in X_T$ and at all times $(x(t),\dot{x}(t)) \in \Gamma_D$. Such a trajectory exists by definition of $X_0 \ni x_0$. Now, consider the map $t \mapsto v(t,x(t))$; by design~\eqref{eq:Lie} of the operator $\cL_y$, constraint~\eqref{con: occupation} ensures that
    \begin{equation}
        0 \geq \cL_y v(t,x(t),\dot{x}(t)) = \left.\frac{d}{dt}\middle[ v(t,x(t)) \right] \label{eq:decrease} \tag{*}
    \end{equation}
    $$ \text{so that} \quad w(x_0) \stackrel{\eqref{con:vw}}{\geq} 1 + v(0,x_0) \stackrel{\eqref{eq:decrease}}{\geq} 1 + v(T,x(T)) \stackrel{\eqref{con:terminal}}{\geq} 1, $$
    i.e. $x_0 \in \hat{X}_0(w)$. Eventually,~\eqref{eq:cv} follows from  $$ \!\!\! \mathrm{vol}(\hat{X}_0(w) \! \setminus \! X_0) \! = \! \mathrm{vol}(\hat{X}_0(w)) \! - \! \mathrm{vol}(X_0) \! \leq \!\! \int \!\! w(x) dx \! - \! w^\star \vspace*{-2em}$$
\end{proof}}

In practice, these converging approximations are computed \emph{via} the \ac{SoS} hierarchy{: certificates $(v,w)$ are parame-trised as $v(t,x) = c_v^\top \varphi(t,x)$ and $w(x) = c_w^\top \psi(x)$, where $\varphi(t,x)$ and $\psi(x)$ are polynomial vectors, and $c_v,c_w$ are real vectors; inequalities are parametrised by \ac{LMI}, for instance constraint~\eqref{con:wgeq} reads $w(x) = \psi(x)^\top Q_1 \psi(x) + g_X(x) \psi(x)^\top Q_g \psi(x)$, $Q_1, Q_g \succeq 0$. See Appendix~\ref{app:sos} for details on the SoS hierarchy.}

\subsection{Toy example revisited}
{The numerical example introduced in~\eqref{eq:1D_example_sys} is reused with dataset $D_1$. In figure~\ref{fig:1D-3p-LP}, we compare the outer approximation $\hat{X}_0(w)$ (in blue, obtained from \ac{LP}~\eqref{eq:LP}) with $X_0(f_{best,D_1})$ (in red
).} This provides a good approximation of the real \ac{FH-RoA}, as 
$w(x)$ closely mimics its indicator function and reaches $w(x)=1$ along the set's boundary.

\begin{figure}[htbp]
    \centering
    \includegraphics[width=0.67\linewidth]{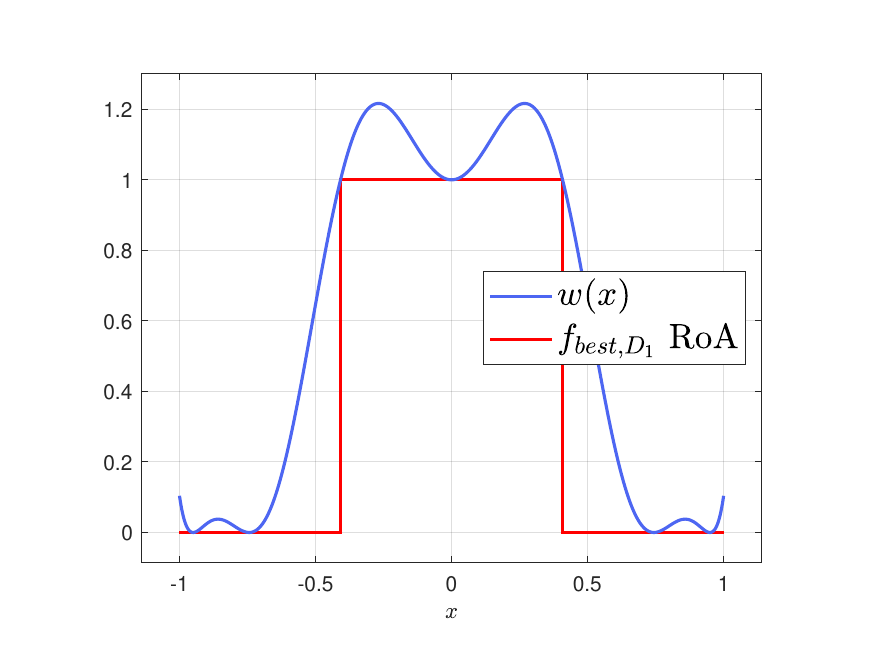}
    \caption{\ac{LP} based \ac{FH-RoA} approximation with 3 data points.}
    \label{fig:1D-3p-LP}
\end{figure}

Applying the same optimisation procedure with dataset $D_2$ produces a narrower outer approximation $\hat{X}_0(w)$ of the \ac{FH-RoA} $X_0(f)$ by approaching the new $X_0(f_{best,D_2})$ (indicator function in red), as illustrated in figure~\ref{fig:1D-5p-LP}. This highlights how the available data affects the accuracy of the approximation, a point further investigated in section~\ref{sec:num}.

\begin{figure}[htbp]
    \centering
    \includegraphics[width=0.67\linewidth]{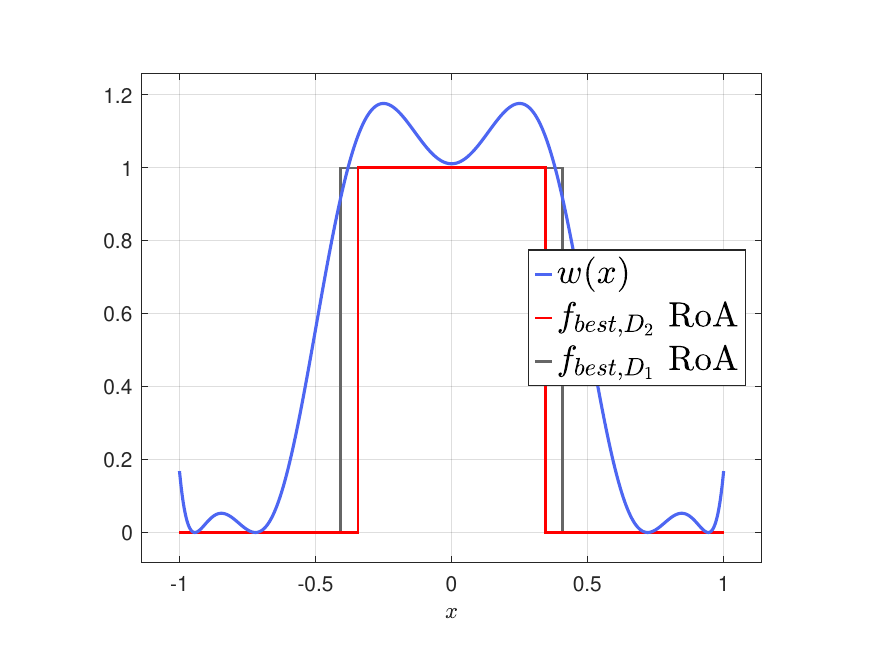}
    \caption{\ac{LP} based \ac{FH-RoA} approximation with 5 data points}
    \label{fig:1D-5p-LP}
\end{figure}

\subsection{Worst-case inner approximations}

As highlighted in Remark~\ref{rem:best}, the above framework concerns the best-case RoA outer approximation. 
The \emph{worst-case \ac{FH-RoA} inner approximation} is relevant for data-driven analysis, \emph{i.e.}, requiring that \emph{all} dynamical systems consistent with the data steer $x_0$ to $X_T$ under state constraints. Below, we provide a procedure to transition from the former to the latter, in the spirit of~\cite{KordaNOLCOS13}; though that reference considers closed-loop systems, whereas we retain the uncertainty set inclusion $\dot{x} \in F_D(x)$.

Consider the \ac{FH-RoA} $\tilde{X}_0$ obtained by replacing $X_T$ with $X_T^c$ in~\eqref{eq:dataROA}, as well as, for any $\tau \in (0,T)$, the \ac{FH-RoA} $\tilde{X}_\tau^{\partial}$ obtained by replacing $T$ with $\tau$ and $X_T$ with $\partial X$ in~\eqref{eq:dataROA}. Then, one can straightforwardly compute outer approximations of $\tilde{X}_0$ and $\tilde{X}_\tau^\partial$ doing the same replacements in~\eqref{eq:LP}. Moreover, defining the \emph{worst-case} \ac{FH-RoA}:
\begin{subequations}
\begin{equation} \label{eq:worst}
    X_0^\star \triangleq \left\{x_0 \in X \middle| \begin{array}{l}
        \forall x(\cdot) \in W^{1,\infty}(0,T)^n \quad \text{s.t.} \\[0pt]
        \forall t \in [0,T], \ \dot{x}(t) \in F_D(x(t)) \\[0pt]
        \text{and} \quad x(0) = x_0, \quad \text{it holds} \\[0pt]
        \forall t \in [0,T], \ x(t) \in X \\[0pt]
        \text{and} \quad x(T) \in X_T 
    \end{array}\right\},
\end{equation}
one gets the following complement formula:

\vspace{-1em}
\begin{align} 
(X_0^\star)^c & = \left\{ x_0 \in X \; \middle| \begin{array}{l}
        \exists x(\cdot) \in W^{1,\infty}(0,T)^n \quad \text{s.t.} \\[0pt]
        \forall t \in [0,T], \quad \dot{x}(t) \in F_D(x(t)), \\[0pt]
        \quad x(0) = x_0 \quad \text{and} \\[0pt]
        \left[\begin{array}{l}
            \exists \tau \in (0,T); \; x(\tau) \in \partial X \\[0pt]
            \qquad \text{or} \quad x(T) \in X_T^c
        \end{array} \right]
    \end{array} \right\} \notag \\[0pt] 
    & = \bigcup_{0<\tau<T}^{\vphantom{m}} \tilde{X}_\tau^\partial \cup \tilde{X}_0,
\label{eq:compl}  
\end{align}
\end{subequations}
so that outer approximations of $(X_0^\star)^c$ can be deduced from a minimising sequence for the following problem:

\vspace{-1em}
\begin{subequations} \label{eq:innerLP}
\begin{align}
    \hspace{-1em} \underset{\substack{v \in C^1([0,T]\times X) \\[0pt] w \in C^0(X)}}{\mathrm{minimise}} & \int w(x) \; dx \\[0pt]
    \text{s.t.} \qquad & \cL_y v(t,x,y) \leq 0 & \hspace{-2em} \forall (t,x,y) \in [0,T]\times \Gamma_D \label{con:occupin} \\[0pt]
    & w(x) \geq 0 & \forall x \in X \label{con:wpos} \\[0pt]
    & w(x) \geq v(0,x) + 1 & \forall x \in X \\[0pt]
    & v(T,x) \geq 0 & \forall x \in {
    X \setminus X_T} \\[0pt]
    & 
    v(t,x) \geq 0 & 
    \hspace{-2em} \forall (t,x) \in [0,T]\times \partial X \label{con:boundary}
\end{align}
\end{subequations}

Since the complements of outer approximations of $(X_0^\star)^c$ are inner approximations of $X_0^\star$, and similarly to 
Lemma~\ref{lem:cv}, any $(v,w)$ feasible for~\eqref{eq:innerLP} is such that:
\begin{align} \label{eq:certin}
    X_0^\star 
    & {
    \; \supset \; } \{x \in X \mid w(x) {
    \; \leq \; } 1\} \triangleq
    \check{X}_0^\star,
\end{align}
then the moment-SoS hierarchy can be implemented to compute such inner approximations $\check{X}_0^\star$.

{\begin{theorem} \label{thm:convergence}
    For a given dataset $D = \{(x_i, f(x_i)\} \subset \R^{2n}$, let $X_0^\star(D)$ be the corresponding worst-case data-based \ac{FH-RoA} as per~\eqref{eq:worst}. Then, for all $N \in \N$ there exists a dataset $D_N = \{(x_i,f(x_i))\}_{i=1}^N \subset \R^{2n}$ such that:
    \begin{subequations}
        \begin{align}
            & \forall N \geq 1, \quad X_0^\star(D_N) \subset X_0(f) \label{eq:best-case} \\[0pt]
            \text{and} \quad & \mathrm{vol}\left(X_0(f) \setminus X_0^\star(D_N) \right) \underset{N\to\infty}{\longrightarrow} 0. \label{eq:dataconv}
        \end{align}
    \end{subequations}
\end{theorem}
\begin{proof}
    See appendix~\ref{app:proof}.
\end{proof}}

{
\begin{remark}
    The difference between this framework and the method proposed in~\cite{TacchiTAC25} is twofold: here, the optimal solutions are proved to converge towards the actual \ac{FH-RoA}, while~\cite{TacchiTAC25} only guarantees convergence of certificates towards a Lyapunov function; plus, this framework uses polynomial certificates and \ac{SoS} programming while the latter is based on piecewise affine Lyapunov functions and second order cone programming. While the uncertainty modeling is close, the methods and the guarantees they yield are completely different.
\end{remark}
}

\section{Numerical results}\label{sec:num}

All the simulations were conducted on a laptop with 32GB of memory and an AMD Ryzen 7 7735U processor. The algorithm was coded using YALMIP~\cite{yalmip}, with MOSEK~\cite{mosek} being the selected solver.

\subsection{Analysis of data position's influence}
In this part, the influence of the position of data points on the outer approximation (\ac{LP}~\eqref{eq:LP}) of the \ac{FH-RoA} is explored. For this, we work with the same 1D system defined in~\eqref{eq:1D_example_sys} on $X=[-1,1]$, and we keep $X_T = [-0.25,0.25]$ for $T=1s$. {We consider the dataset $D(p) = D_1 \cup \{(-p;f(-p)),(p;f(p))\}$ of 5 data points, of which three are fixed (in $D_1$) and two vary in position according to $p\in [0.1,0.5]$}. In Figure~\ref{fig:1D-variation}, and in light of the uncertainty set representation of Figure~\ref {fig:1D-3p-example}, the more {the two data points of positions $p$ and $-p$} restrain $F_D(X)$, the closer the outer approximation is to the \ac{FH-RoA} of the real function.

\begin{figure}[htbp]
    \centering
    \begin{subfigure}{.49\linewidth}
        \includegraphics[width=\linewidth]{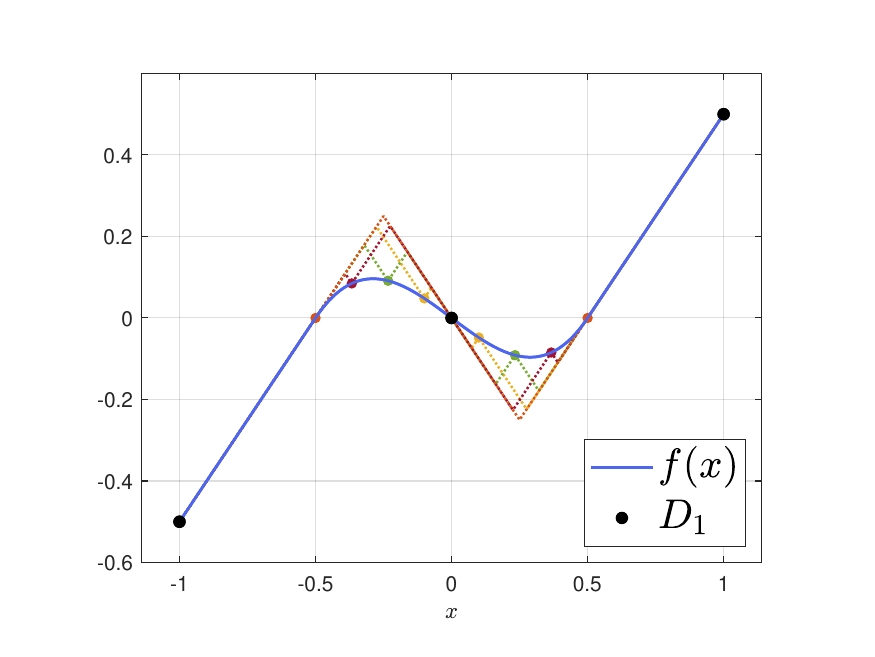}
        \caption{\centering all variations of~$D(p)$ and~their corresponding $f_{best}$}
        \label{fig:1D-variation-data}
    \end{subfigure}
    \begin{subfigure}{.49\linewidth}
        \includegraphics[width=\linewidth]{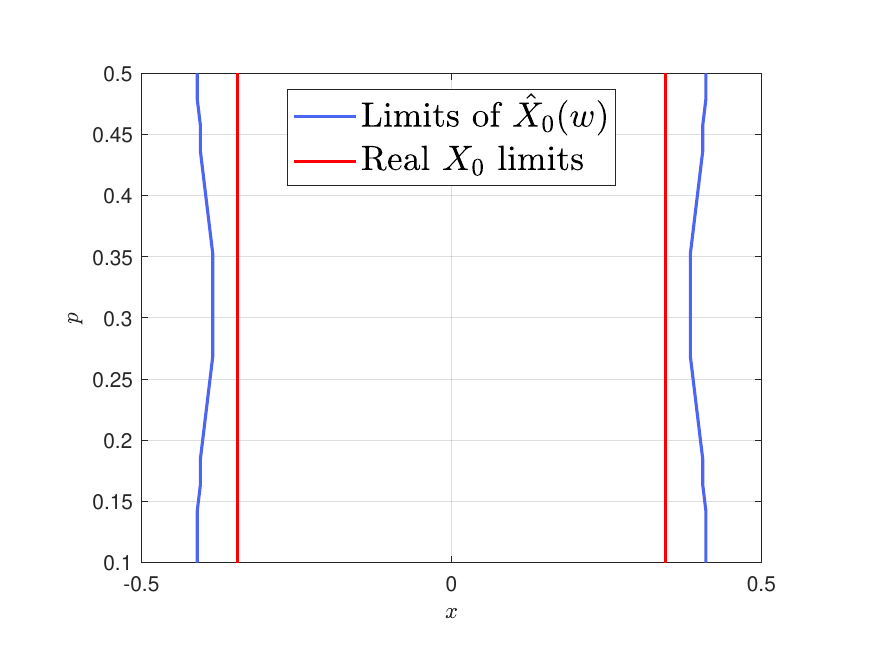}
        \caption{\centering variation of $\hat{X}_0(w)$ according to~$D(p)$}
        \label{fig:1D-variation-RoA}
    \end{subfigure}
    \caption{Simulation using the variable dataset $D(p)$}
    \label{fig:1D-variation}
\end{figure}

\subsection{Numerical example in $\R^2$}
The SoS reformulations of the \ac{LP}~\eqref{eq:LP} is solved for a dynamical system with a Lipschitz constant $M=1$ defined on $x \in X=[-0.8,0.8]^2\subset\R^2$ as follows:
\begin{equation}
    \dot{x} = 
        \begin{cases}
        2x(\|x\|^2-0.5^2) \quad \text{if} \quad \|x\| \leq 0.5 \\[0pt]
        x\left(1-\frac{1}{2\|x\|}\right) \quad \text{else}.
        \end{cases}
\end{equation}

We seek to find the \ac{FH-RoA} of the system with $x \in X_T=\{x\in X \ | \ (0.25^2-\|x\|^2) \geq 0\}$ at  time $T=1s$. For this, we fix the degree of the polynomials to $10$, and the dataset is generated randomly in $X$ with $N = 50$. See the results in Figure~\ref{fig:2D-full}. 
\begin{figure}[htbp]
    \centering
    \begin{subfigure}{.49\linewidth}
        \includegraphics[width=\linewidth]{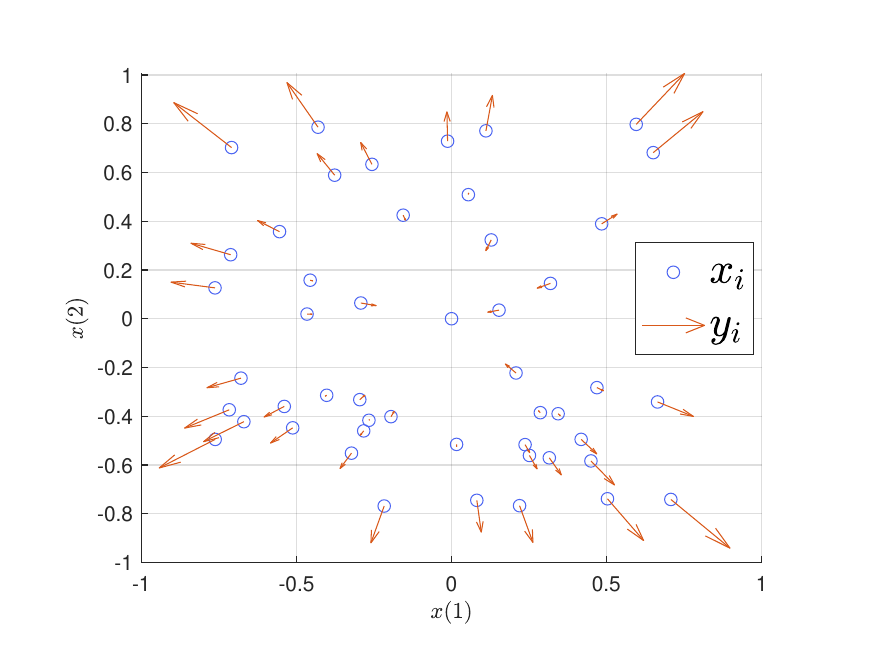}
        \caption{generated dataset}
        \label{fig:2D-full-data}
    \end{subfigure}
    \begin{subfigure}{.49\linewidth}
        \includegraphics[width=\linewidth]{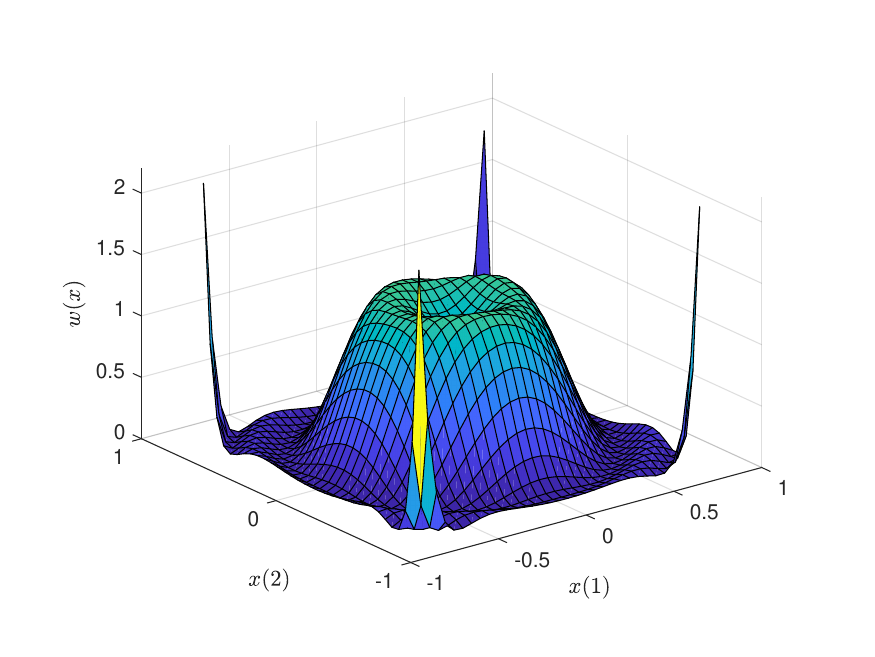}
        \caption{polynomial $w(x)$}
        \label{fig:2D-full-w}
    \end{subfigure}
    \begin{subfigure}{.49\linewidth}
        \includegraphics[width=\linewidth]{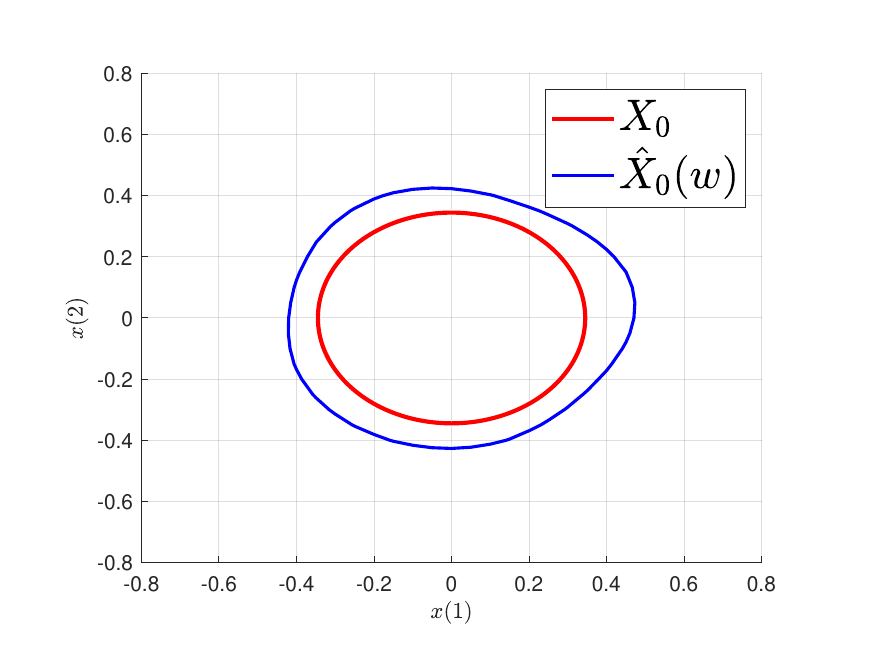}
        \caption{resulting \ac{FH-RoA}}
        \label{fig:2D-full-RoA}
    \end{subfigure}
    \caption{Simulation results for $X_T$}
    \label{fig:2D-full}
\end{figure}

In \ref{fig:2D-full-RoA}, the obtained approximation $\hat{X}_0(w)$ in blue approaches the actual \ac{FH-RoA} of the system in red. The number of data points and the chosen degree of the polynomials influence the algorithm's performance, the approximation's accuracy, and can lead to a better approximation of the true function \ac{FH-RoA}. The dataset shown in Figure~\ref{fig:2D-full-data} influences the shape of the obtained $\hat{X}_0(w)$ through its effect on the derived best-case function. By providing the \ac{LP} with a different dataset, we can generate a new $\hat{X}_0(w)$ that addresses the shape disfigurement of the previous result without increasing the number of data points, as illustrated in figure~\ref{fig:2D-full-V2}.

\begin{figure}[htbp]
    \centering
    \begin{subfigure}{.49\linewidth}
        \includegraphics[width=\linewidth]{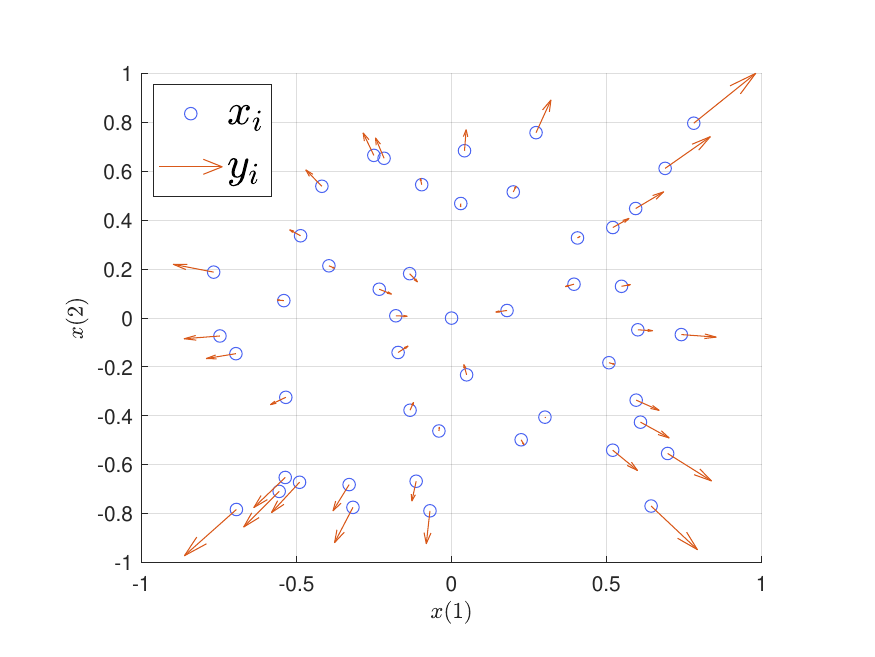}
        \caption{new generated dataset}
        \label{fig:2D-full-data-V2}
    \end{subfigure}
    \begin{subfigure}{.49\linewidth}
        \includegraphics[width=\linewidth]{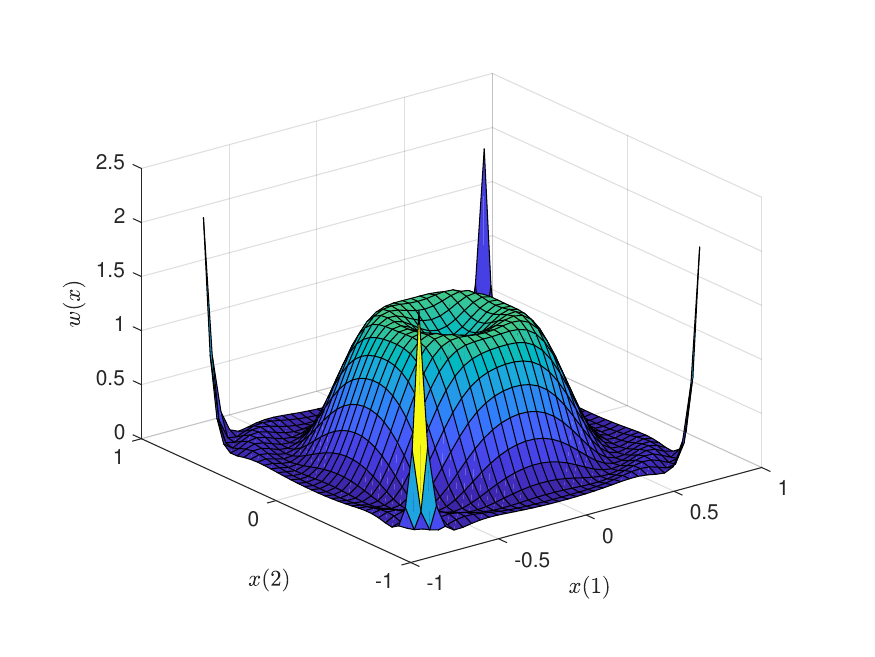}
        \caption{polynomial $w(x)$}
        \label{fig:2D-full-w-V2}
    \end{subfigure}
    \begin{subfigure}{.49\linewidth}
        \includegraphics[width=\linewidth]{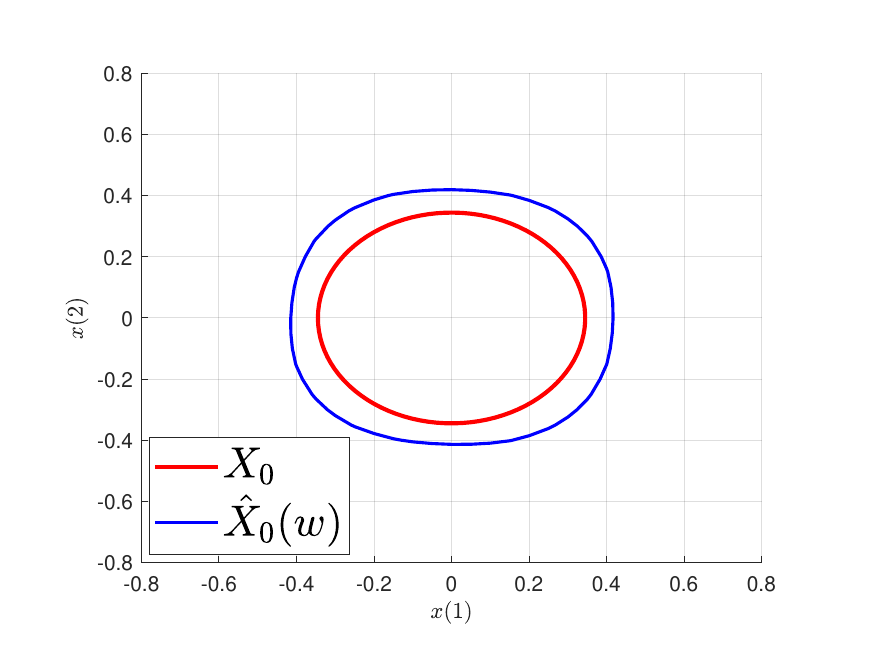}
        \caption{resulting \ac{FH-RoA}}
        \label{fig:2D-full-RoA-V2}
    \end{subfigure}
    \caption{Alternative simulation results for $X_T$}
    \label{fig:2D-full-V2}
\end{figure}

To test the validity of the algorithm for non-symmetric \ac{FH-RoA}s, the previous $X_T$ is changed into a non-symmetric set for the same time $T=1s$: 
\begin{equation*}
    X_T'= \left\{x\in X \  \middle | \ \begin{array}{c}
    (0.25^2-x^2) \geq 0 \\[0pt] x_{(1)}+x_{(2)} \geq 0 
    \end{array} \right\}
\end{equation*}

The simulation results from such a change are shown in Figure~\ref {fig:2D-half}. The shape of the approximating set $\hat{X}_0$ in Figure~\ref{fig:2D-half-RoA} changes to accommodate the new target set $X_T'$ so that it still approaches the real \ac{FH-RoA}.
\begin{figure}[htbp]
    \centering
    \begin{subfigure}{.49\linewidth}
        \includegraphics[width=\linewidth]{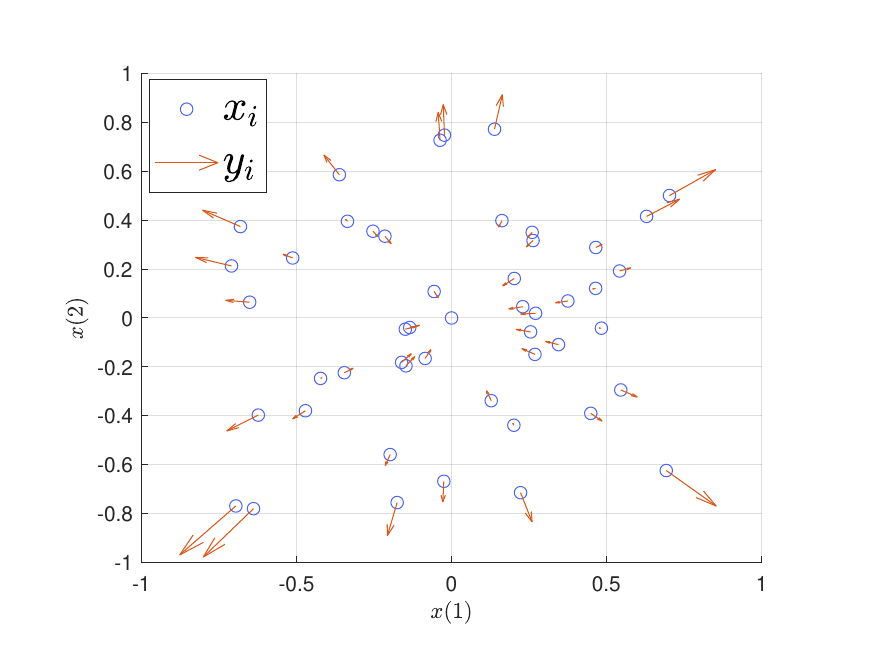}
        \caption{generated dataset}
        \label{fig:2D-half-data}
    \end{subfigure}
    \begin{subfigure}{.49\linewidth}
        \includegraphics[width=\linewidth]{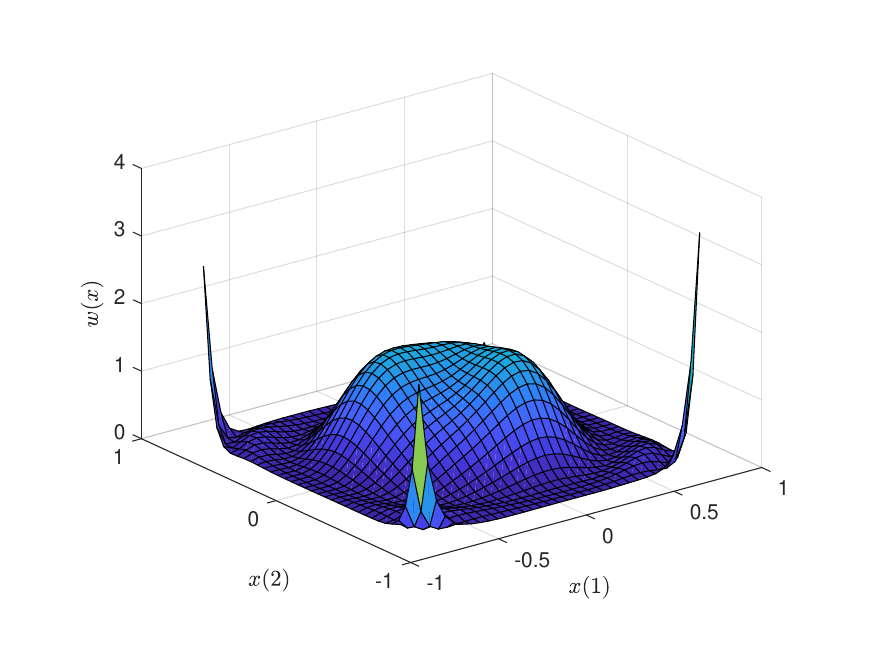}
        \caption{polynomial $w(x)$}
        \label{fig:2D-half-w}
    \end{subfigure}
    \begin{subfigure}{.51\linewidth}
        \includegraphics[width=\linewidth]{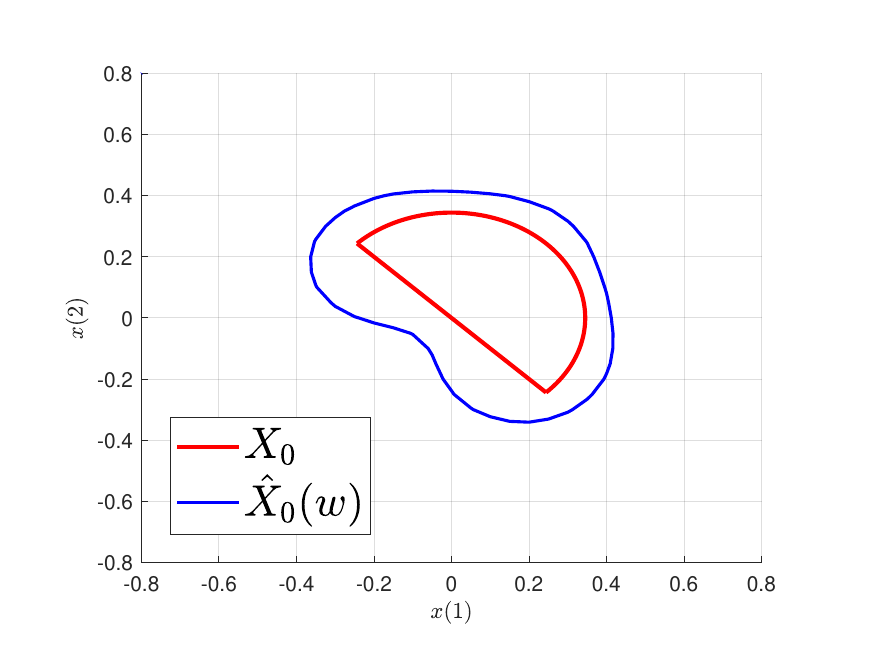}
        \caption{resulting \ac{FH-RoA}}
        \label{fig:2D-half-RoA}
    \end{subfigure}
    \caption{Simulation results for $X_T'$}
    \label{fig:2D-half}
\end{figure}

\subsection{Computational challenges and proposed resolution}
This proposed approach offers several advantages, including that $f$ does not need to be polynomial or known like it is supposed in \cite{KordaTAC14,KordaNOLCOS13}. It is also proven to converge, as shown in Appendix~\ref{app:proof}. It faces, however, some numerical problems when the value of $M$ is high or when $T\gg 1s$. In the first case, a high overestimate of $M$ is not informative enough for the \ac{LP} to approach the \ac{FH-RoA} 
{(see the influence of $M$ on the convergence proof in Appendix~\ref{app:proof}). When $T \gg 1s$, even model-based approaches are hampered by numerical conditioning issues in the SDP approximations: this stems from the geometric properties of Putinar's Positivstellensatz as detailed in~\cite{schlosser2026rates}; this is worse in the data-based setting where model approximation errors accumulate along the time horizon. 
} 

In future work, we aim to address this issue by introducing a partitioning strategy for the set $X$, inspired by the approach in \cite{CibulkaLCSS21,TacchiTAC25}. This method is expected to enhance the results by introducing local Lipschitz constants and data constraints for each cell in the set. It would also allow for a lower degree of polynomials with no significant drawbacks to accuracy. 

Another potential improvement to the results can be achieved by changing the polynomial basis of the \ac{SoS} formulation. An \ac{SoS} decomposition expresses a polynomial as a quadratic form in a chosen polynomial basis with a semi-definite positive weight matrix, which can be checked using \ac{LMI}. The polynomial basis influences the accuracy of the numerical approximation as shown in~\cite{henrion2009approximate}. {This idea has been used in~\cite{AAA2017DSOS} to drastically improve the scalability of \ac{SoS} programming, which is known to be its main practical weakness.} 

\section{Conclusion}

This paper introduces a novel data-driven approach to approximate the \ac{FH-RoA} via the moment-SoS hierarchy. This methodology effectively provides outer approximations of the real \ac{FH-RoA} using a dataset and a Lipschitz constant of the unknown system dynamics. In contrast to model-based approaches, using data circumvents the limitation of the system's polynomiality. While the results are promising, further improvements are needed to address the scalability issue arising from numerical challenges. Future work includes introducing set partitioning to address the scalability issue, further developing the worst-case inner approximation part, investigating the influence of data on conservatism, and testing a polynomial basis for the \ac{SoS} formulation other than monomials.

\bibliographystyle{ieeetr}
\bibliography{RefIFAC26}

\appendix

\section{The sums-of-squares hierarchy} \label{app:sos}

Introducing the cone 
$$\Sigma[z] = \{p_1(z)^2+\ldots+p_K(z)^2 \mid K \in \N, p_1,\ldots,p_K \in \R[z]\}$$
of polynomial \ac{SoS} in the variable $z \in \{x,(t,x,y)\}$, the univariate polynomial $q(t) = (T-t)t$, the size $N$ vector
$$ \gamma(x,y) = (M^2\|x-x_i\|^2 - \|y-y_i\|^2)_{i=1}^N $$
and, for any vector of polynomials $g(z) \in \R[z]^p$ in the variable $z$, the quadratic module defined as
\begin{equation}
    \cQ(g(z)) = \left\{\sigma(z) + s(z)^\top g(z) \middle | \begin{array}{l}
        \sigma(z) \in \Sigma[z], \\[0pt] s(z) \in \Sigma[z]^p
    \end{array} \right\},
\end{equation}
the \ac{LP}~\eqref{eq:LP} can be recast into the following:

\begin{subequations} \label{eq:SoS}
\begin{align}
    \underset{\substack{v \in \R[t,x] \\[0pt] w \in \R[x]}}{\mathrm{minimise}} & \int w(x) \; dx \\[0pt]
    \text{s.t.} \qquad & \hspace{-2em} -\cL_y v(t,x,y) \in \cQ\left(\vphantom\sum q(t), g_X(x), \gamma(x,y)\right)  \label{sos:occupation} \\[0pt]
    & w(x) \in \cQ(g_X(x)) \\[0pt]
    & w(x) - v(0,x) - 1 \in \cQ(g_X(x)) \\[0pt]
    & v(T,x) \in \cQ(g_T(x)) 
\end{align}
\end{subequations}

\begin{proposition}
    Problem~\eqref{eq:SoS} is equivalent to problem~\eqref{eq:LP} in the sense that any $(v,w)$ feasible for~\eqref{eq:SoS} is feasible for~\eqref{eq:LP} and minimising sequences $(v_\epsilon,w_\epsilon)_\epsilon$ for~\eqref{eq:LP} can be approximated by minimising sequences $(\hat{v}_\epsilon,\hat{w}_\epsilon)_\epsilon$ for~\eqref{eq:SoS}.
\end{proposition}
\begin{proof}
    Since by design any $\sigma(z) \in \Sigma[z]$ is non-negative, it is clear that any $q(z) \in \cQ(g(z))$ is non-negative on $\{z \mid g(z) \geq 0\}$. Hence, for $\text{x} \in \{\text{b, c, d, e}\}$, (\ref{eq:SoS}x) implies (\ref{eq:LP}x) and feasibility for~\eqref{eq:SoS} implies feasibility for~\eqref{eq:LP}.

    Next, let $(v_\epsilon,w_\epsilon)$ be a minimising sequence for~\eqref{eq:LP} that is strictly feasible (meaning that all the inequality constraints are strict). Since by Assumption~\ref{asm: algebra} $X$ is compact, the Weierstrass theorem yields a polynomial approximation $(\hat{v}_\epsilon, \hat{w}_\epsilon)$ of $(v_\epsilon, w_\epsilon)$ that is also a strictly feasible minimising sequence for~\eqref{eq:LP}. Eventually, we state Putinar's Positivstellensatz~\cite{Putinar93}: under mild assumptions (that are implied by Assumption~\ref{asm: algebra} in our case), if $q(z) > 0$ on $\{z \mid g(z) \geq 0\}$, then $q(z) \in \cQ(g(z))$. Then, since the $(\hat{v}_\epsilon, \hat{w}_\epsilon)$ are strictly feasible for~\eqref{eq:LP}, they are also feasible for~\eqref{eq:SoS}. Since~\eqref{eq:SoS} is a strengthening of~\eqref{eq:LP}, any minimising sequence of~\eqref{eq:LP} that is also feasible for~\eqref{eq:SoS} is a minimising sequence of~\eqref{eq:SoS}.
\end{proof}

The quadratic module $\cQ(g(z))$ can be truncated into a bounded degree quadratic module $\cQ_d(g(z))$ where $\sigma(z)$ and each term $s_j(z)\cdot g_j(z), j \in [p]$ of $s(z)^\top g(z)$ have degree at most $d \in \N$. Such a truncated quadratic module has the particularity of being \ac{LMI}-representable (see e.g. Prop. 2.1 in ~\cite{Lasserre10}). Then, the \ac{SoS} hierarchy consists of the sequence of problems obtained by replacing $\cQ$ with $\cQ_d, d \in \N$ in~\eqref{eq:SoS}; since the $\cQ_d$ are nested (in the sense that $\cQ_d(g(z)) \subset \cQ_{d+1}(g(z)) \subset\cQ(g(z))$), this produces a monotonic sequence of convex, finite-dimensional strengthenings of~\eqref{eq:SoS}.

Problem~\eqref{eq:SoS} has a linear cost and any of its finite-dimensional truncations has a non-empty, compact feasible set, so they all have an optimal solution. Moreover, since $\cQ(g(z)) = \cup_{d\in\N} \cQ_d(g(z))$, these optimal solutions form a minimising sequence for~\eqref{eq:SoS} and hence for~\eqref{eq:LP}, yielding a converging sequence of outer approximations of $X_0$ as described in~\eqref{eq:cv} when the degree bound $d$ goes to infinity.

The main drawback of approximating~\eqref{eq:LP} with truncated instances of~\eqref{eq:SoS} is the computational burden: at degree $d$, the involved \ac{LMI} have size growing like $\binom{2n+d+1}{d}$, i.e. combinatorially in the state dimension as well as in the degree of the involved polynomials. Furthermore, the \ac{LMI} representation of $\Sigma[z]$ heavily depends on the choice of the polynomial basis for $\R[z]$, which can have significant effects on the numerical behaviour of the resulting \ac{SDP} problems.

\section{Convergence of data-based regions of attraction}\label{app:proof}

{The proof of Theorem~\ref{thm:convergence} is as follows:}

\begin{proof}
\eqref{eq:best-case} holds by definition of the \textit{worst-case} RoA $X_0^\star(D)$: any of its elements only initialises trajectories that reach the target under dynamics compatible with the dataset $D$. By design, $f$ is compatible with $D$, so that its trajectories initialized in $X_0^\star(D)$ indeed reach the target, i.e. $X_0^\star(D) \subset X_0(f)$. Then, we prove~\eqref{eq:dataconv}. Both $X_0(f)$ and $X_0^\star(D)$ are respectively approximated by inside by:
$$\check{X}_0(w_0) = \{x \in X \mid w_0(x) \leq 1\},$$ and: $$\check{X}_0(w_D) = \{x \in X \mid w_D(x) \leq 1\},$$
with $w_D$ feasible for the data-based \ac{LP} in~\eqref{eq:innerLP}, and $w_0$ feasible for the model-based \ac{LP} as described in~\cite{KordaTAC14}:
\begin{equation} \label{eq:model}
\begin{array}{clr}
    \mathrm{minimise} & \displaystyle \int_X w_0(x) \; dx & \\[0pt]
    \mathrm{s.t.} 
    & \cL_{f}v_0(t,x) \leq 0 & \forall (t,x) \in [0,T]\times X \\[0pt]
    & w_0(x) \geq 0 \hspace{-2em} & \forall x \in X \\[0pt]
    & w_0(x) \geq v_0(0,x) + 1 \hspace{-2em} & \forall x \in X \\[0pt]
    & v_0(T,x) \geq 0 & \forall x \in X \setminus X_T \\[0pt]
    & v_0(t,x) \geq 0 & \forall (t,x) \in [0,T] \times \partial X 
\end{array}
\end{equation}
where $\cL_{f}$ denotes the operator: $$v(t,x) \longmapsto \partial_t v(t,x) + f(x)^\top \partial_x v(t,x).$$ 
A feasibility-preserving direction for \eqref{eq:model} is defined by $v_\delta(t,x) = T-t$ and $w_\delta(x) = T \geq 0$, with:
\begin{align*}
    w_\delta(x) - v_\delta(0,x) = 0  & \qquad \forall x \in X \\[0pt]
    v_\delta(T,x) = 0 & \qquad \forall x \in X \setminus X_T \\[0pt]
    v_\delta(t,x) \geq T-T = 0 & \qquad \forall (t,x) \in [0,T] \times \partial X \\[0pt]
    \cL_{f}v_\delta(t,x) = -1 & \qquad \forall (t,x) \in [0,T]\times X
\end{align*}

Moreover, the following holds:
$$ \int_X w_\delta(x) \; dx = T \mathrm{vol}(X) \triangleq
C. $$
From~\cite{KordaNOLCOS13}, we know that, for $\varepsilon > 0$, there exists an $\nicefrac{\varepsilon}{2}$-optimal feasible solution $(v_\varepsilon, w_\varepsilon)$ for~\eqref{eq:model} in the sense that: 
$$ \check{X}_0(w_\varepsilon) \subset X_0(f), \quad \text{and}$$
\begin{align*}
\mathrm{vol}\left(X_0(f) \! \setminus \! \check{X}_0(w_\varepsilon)\right) & \leq \! \int_X \! w_\varepsilon(x) dx - \mathrm{vol}\left(X \! \setminus \! X_0(f)\right) 
\leq \frac{\varepsilon}{2}.
\end{align*}

Let $\eta = \nicefrac{\varepsilon}{2C} > 0$, and define:
$$ v_\eta(t,x) = v_\varepsilon(t,x) + \eta \cdot v_\delta(t,x),$$ 
$$ w_\eta(x) = w_\varepsilon(x) + \eta \cdot w_\delta(x) \geq T\cdot \eta > 0. $$
We check that, on the appropriate domains, it holds:
\begin{align*}
    w_\eta(x) - v_\eta(0,x) & = \underset{\geq 1}{\underbrace{w_\varepsilon(x) - v_\varepsilon(0,x)}} \\[0pt] & \hspace{-1em} + \underset{= 0}{\underbrace{(w_\delta(x) - v_\delta(0,x))\cdot\eta}} \geq 1 \\[0pt]
    v_\eta(T,x) & = \underset{\geq 0}{\underbrace{v_\varepsilon(T,x)}} + \underset{=0}{\underbrace{\eta \cdot v_\delta(T,x)}} \geq 0 \\[0pt]
    v_\eta(t,x) & = \underset{\geq 0}{\underbrace{v_\varepsilon(t,x)}} + \underset{\geq 0}{\underbrace{\eta \cdot v_\delta(t,x)}} \geq 0 \\[0pt]
    \cL_{f} v_\eta(t,x) & = \underset{\leq 0}{\underbrace{\cL_{f}v_\varepsilon(t,x)}} + \underset{= -\eta }{\underbrace{\cL_{f}v_\delta(t,x))\cdot\eta}} 
    \leq - \eta,
\end{align*}
so that $(v_\eta,w_\eta)$ is feasible for \eqref{eq:model}.
We are going to derive a condition on the dataset $D$ for $(v_\eta,w_\eta)$ to be feasible for the data-based problem~\eqref{eq:innerLP}. Notice that by construction, constraints~\eqref{con:wpos}--\eqref{con:boundary} are satisfied by $(v_\eta,w_\eta)$. Hence we look at constraint~\eqref{con:occupin}. Let $t \in [0,T]$, $(x,y) \in \Gamma_D$. We just need to prove:
$$ y^\top \partial_xv_\eta(t,x) \leq f(x)^\top\partial_xv_\eta(t,x) + \eta, $$
since we already know that the RHS is not greater than $-\partial_tv_\eta(t,x)$. Hence, we look for a condition such that:
$$ (y-f(x))^\top \partial_xv_\eta(t,x) \leq \eta. $$
Using the Cauchy-Schwarz inequality, we get:
$$(y-f(x))^\top \partial_xv_\eta(t,x) \leq \|y-f(x))\| \cdot \|\partial_xv_\eta(t,x)\|.$$
Moreover, for any $i \in [N]$, it holds:
$$\|y-f(x)\| \leq \underset{\leq M \|x-x_i\|}{\underbrace{\|y - f(x_i)\|}} \hspace{-.2em} + \hspace{-.2em} \underset{\leq M\|x-x_i\|}{\underbrace{\|f(x_i) - f(x)\|}} \leq 2M \|x-x_i\|$$
where the first braced inequality comes from the definition of $\Gamma_D \ni (x,y)$ and the second comes from the fact that $M$ upper bounds the Lipschitz constant of $f$. Hence, defining:
$$ \rho = \frac{\eta}{2M \|\partial_xv_\eta\|_\infty}, $$
we have just proved that if for all $x \in X$, there exists an $i \in [N]$ with $\|x-x_i\| \leq \rho$, then:
\begin{align*}
    \cL_y v_\eta(t,x) & = \underset{\leq -\eta}{\underbrace{\cL_{f}v_\eta(t,x)}} + (y-f(x))^\top \partial_x v_\eta(t,x) \\[0pt]
    & \leq \|y-f(x)\| \cdot \|\partial_xv_\eta(t,x)\| - \eta \\[0pt]
    & \leq 2M \|x-x_i\| \cdot \|\partial_xv_\eta\|_\infty - \eta \\[0pt]
    & \leq 2M \rho \|\partial_xv_\eta\|_\infty - \eta = 0.
\end{align*}
Since $X$ is compact, it can be covered by a finite number of balls of radius $\rho$. Hence, sampling $f$ at the centres of these balls gives an appropriate dataset $D = \{(x_i, f(x_i))\}_{i=1}^N$ such that for all $x \in X$ there is an index $i \in [N]$ with $\|x-x_i\| \leq \rho$. For this specific dataset, we have proved that $(v_\eta, w_\eta)$ is feasible for both problems~\eqref{eq:model} and \eqref{eq:innerLP}. 
$$ \text{Hence, it holds} \quad \check{X}_0(w_\eta) \subset X_0^\star(D) \subset X_0(f), \quad \text{and} $$
\begin{align*}
    \mathrm{vol}(X_0(f)) \setminus \check{X}_0(w_\eta)) & \leq \int_X w_\eta(x) \; dx - \mathrm{vol}(X \setminus X_0(f)) \\[0pt]
    & = \underset{\leq \nicefrac{\varepsilon}{2}}{\underbrace{\int_X w_\varepsilon(x) \; dx - \mathrm{vol}(X \setminus X_0(f))}} \\[0pt]
    & + \eta \underset{=C}{\underbrace{\int_X w_\delta(x) \; dx}} 
    \leq \nicefrac{\varepsilon}{2} + C\cdot \eta = \varepsilon
\end{align*}
so that $X_0(f)$ is approximated arbitrarily closely (from inside) by a feasible function of the data-based problem~\eqref{eq:innerLP}, with no knowledge of the full dynamics.
\end{proof}

\begin{remark}
    What we actually proved is that, up to sufficiently informative datasets $D_N$ and high relaxation orders in the hierarchy (see Appendix~\ref{app:sos}), one can approximate the true \ac{FH-RoA} from inside arbitrarily closely, by solving \ac{SDP} problems. 
    Although less meaningful, the result also stands for best-case outer approximation~\eqref{eq:dataROA} with:
    $$ \mathrm{vol}\left(X_0(D_N) \setminus X_0(f) \right) \underset{N\to\infty}{\longrightarrow} 0, $$
    and the proof is identical to the above.
\end{remark}

{
\begin{remark}
    The above proof highlights how crucial the validity and quality of the Lipschitz upper bound $M$ is. If $M$ underestimates the Lipschitz constant, then $f(x)$ can take values outside of $F_D(x)$, hence hampering the upper bound on $\|y - f(x)\|$ in the proof. If $M$ is very large, then the radius $\rho$ in the proof becomes very small, hence asking for a much larger dataset to obtain the appropriate covering of $X$. In short, the larger $M$, the more conservative the \ac{FH-RoA} approximation. Methods are proposed to estimate and validate $M$ from data in~\cite{martin-data-driven-2024,makdesi2021efficient}. In this paper, we only assume that $M$ is known in advance.
\end{remark}
}

\end{document}